\documentclass{article}
\usepackage{amsfonts}
\usepackage{amsmath}
\usepackage{amsthm}

\renewcommand{\phi}{\varphi}
\title{The expressiveness of MTL with counting}
\author{Paul Hunter}
\date{May 2012}

\newtheorem{lemma}{Lemma}
\newtheorem{theorem}{Theorem}

\begin{document}
\maketitle
\begin{abstract}
It is well known that MTL with integer endpoints is unable to express all of monadic first-order logic of order and metric (FO$(<,+1)$). Indeed, MTL is unable to express the counting modalities $C_n$ that assert a properties holds $n$ times in the next time interval. We show that MTL with the counting modalities, MTL+C, is expressively complete for FO$(<,+1)$. This result strongly supports the assertion of Hirshfeld and Rabinovich that Q2MLO is the most expressive decidable fragments of FO$(<,+1)$.
\end{abstract}

\section*{Preliminaries}
\subsection*{MTL+C}
We are interested in MTL (with past operators) plus
\begin{itemize}
\item Counting modalities $C_n$, $\overline{C_n}$ for $n \in \mathbb{N}$, and
\item Punctuality modalities $\Diamond_{=1}$, $\overline{\Diamond_{=1}}$.
\end{itemize}
Intuitively $C_n(\varphi)$ holds if $\varphi$ holds in at least $n$ distinct times in the next (strict) unit time interval, and $\Diamond_{=1}\varphi$ holds if $\varphi$ holds in exactly one time unit from now.  $\overline{C_n}$ and $\overline{\Diamond_{=1}}$ are the temporal duals ($n$ times in the previous unit time interval and exactly one time unit in the past respectively).  We call this logic MTL+C.

\subsection*{Q2MLO with punctuality}
It is well known that MTL together with the counting modalities is equivalent to Q2MLO, the first-order theory of linear order with monadic predicates, equipped with the metric quantifier
\[ \exists_z^{z+1} y.\varphi(y,z)\]
that can only be applied to formulas with two free variables (including the one being quantified).

\paragraph{Adding past}
It is clear that by including temporal dual operators, MTL with past and counting operators can express ``reverse'' metric quantifiers \textit{viz.}\footnote{Not sure if MTL+counting $\leftrightarrow$ Q2MLO includes past operators}
\[ \exists_{z-1}^{z} y.\varphi(y,z).\]

\paragraph{Adding punctuality}
To capture punctuality, we add to Q2MLO the $+1$ function (technically, the $+1$ relation), to obtain Q2MLO($+1$).  It should be clear that the resulting logic is as expressive as MTL+C.

\section*{Equivalence of bounded FO($<$,$+1$) and Q2MLO($+1$)}
Our aim is to show the following:
\begin{theorem}
For every bounded FO($<$,$+1$) formula $\psi(z)$ there is an equivalent Q2MLO($+1$) formula $\psi'(z)$.
\end{theorem}

\subsection*{Simplified form}
We first show that it suffices to consider FO($<$,$+1$) formulas in a simplified form.

\begin{enumerate}
\item Using stacking we can assume that all quantified variables are in the range $(z,z+1)$ and that $+1$ does not otherwise occur.
\item Using Hodkinson, we can assume the formula is of the form
\[ \psi(z) = \exists x_1 \exists x_2 \cdots \exists x_n \forall y. \varphi(\overline{x},y,z).\]
\item By taking a disjunction over all possible orderings of the $x_i$ we can assume $z<x_1<x_2<\cdots<x_n<z+1$.  
\item We relativize the subformula $\forall y. \varphi(\overline{x},y,z)$ the subintervals $(z,x_1),{x_1},(x_1,x_2)\ldots$:
\[\begin{array}{rccl}
\forall y. \varphi(\overline{x},y,z) &=&&\forall y \in (z,x_1). \varphi_1(\overline{x},y,z)\\
&&\wedge&\forall y \in \{x_1\}. \varphi_2(\overline{x},y,z)\\
&&\wedge&\forall y \in (x_1,x_2). \varphi_3(\overline{x},y,z)\\
&&\wedge&\cdots \\
&&\wedge&\forall y \in (x_n,z+1). \varphi_{2n+1}(\overline{x},y,z).
\end{array}\]
\end{enumerate}
Now each $\varphi_i$ is quantifier-free, and the relativization of $y$ means the binary relations between $\overline{x}$, $y$, and $z$ are all known.  So $\varphi_i$ is a boolean combination of monadic predicates.  By considering $\psi(z)$ as a disjunction over all possible choices of predicate values for $z$ and $\overline{x}$, we can further simplify each $\varphi_i$ to a boolean combination of monadic predicates in $y$.  That is, we need only consider formulas of the form:
\[\begin{array}{rcll}
\psi(z) &=& \exists x_1 \exists x_2 \cdots \exists x_n\:& z < x_1 < x_2 < \cdots < x_n < z+1\\
&&&\wedge\: \forall y \in (z,x_1). \varphi_1(y) \\
&&&\wedge\: \forall y \in \{x_1\}. \varphi_2(y) \\
&&&\wedge\: \forall y \in (x_1,x_2). \varphi_3(y) \\
&&&\wedge\: \cdots\\
&&&\wedge\: \forall y \in (x_n,z+1). \varphi_{2n+1}(y).
\end{array}\]

\subsection*{Equivalence}
Let $\psi(z)$ be a FOMLO formula in the form described above.  For convenience, let $\psi(z,z')$ be the FOMLO formula obtained by replacing (both) occurrences of $z+1$ with the variable $z'$ (so, with abuse of notation, $\psi(z) = \psi(z,z+1)$).  Also for $1\leq i \leq 2n+1$, let
\[\begin{array}{rcll}
\psi_i(z,z') &:=& \exists x_0 \exists x_1 \cdots \exists x_k \:& z =x_0 < x_1 < x_2 < \cdots < x_k = z'\\
&&&\wedge\: \forall y \in (x_0,x_1). \varphi_1(y) \\
&&&\wedge\: \forall y \in \{x_1\}. \varphi_2(y) \\
&&&\wedge\: \cdots\\
&&&\wedge\: \forall y \in I_i. \varphi_{i}(y),
\end{array}\]
where $k = \lceil \frac{i}{2} \rceil$ and $I_i = \{x_k\}$ if $i$ is even and $I_i = (x_{k-1},x_k)$ if $i$ is odd.   That is, $\psi_i(z,z')$ is the formula obtained by considering the first $i$ relativized conjuncts in $\psi(z,z')$ (with some book-keeping to simplify the presentation).

\begin{lemma} 
$\psi(z)$ is equivalent to
\[ \psi'(z):=\theta_1(z) \wedge \theta_2(w) \wedge (w=z+1)\]
where
\begin{eqnarray*}
\theta_1(z) &=& \forall u \in (z,z+1).\bigvee_{i=1}^{2n+1} \psi_i(z,u)\textrm{ and}\\
\theta_2(w) &=& \exists u \in (w-1,w).\psi(u,w).
\end{eqnarray*}

\end{lemma}
\begin{proof}
$\psi(z) \Rightarrow \psi'(z)$.  Let $x_1, \ldots, x_n \in (z,z+1)$ be witnesses for the existential quantifiers in $\psi$, and let $x_0=z$ and $x_{n+1}=z+1$.  From the definition of $\psi_i$, if $u \in (x_i,x_{i+1})$ (for $0 \leq i \leq n$) then $\psi_{2i+1}(z,u)$ holds.  Further, if $u = x_i$ (for $1 \leq i \leq n$) then $\psi_{2i}(z,u)$ holds.  Thus $\theta_1$ is satisfied for all $u \in (z,z+1)$.  Any $u \in (z,x_1)$ is a witness for $\theta_2(z+1)$, and as $x_1\leq z+1$, $u \in (w-1,w)$ where $w=z+1$.  Thus $\theta_2$ holds when $w=z+1$.  Thus $\psi'(z)$ is satisfied.

\vspace*{1ex}

$\psi'(z) \Rightarrow \psi(z)$.  Note that if $\psi_r(z,u)$ holds for $u$ arbitrarily close to $z+1$ then $\psi_r(z,z+1)$ holds.  In particular, if $\psi_{2n+1}(z,u)$ holds for $u$ arbitrarily close to $z+1$ then we are done.  As $\bigvee_{i=1}^{2n+1} \psi_i(z,u)$ holds for all $u \in (z,z+1)$, there is some $r$ such that $\psi_r(z,u)$ holds arbitrarily close to $z+1$.  It follows that $\psi_r(z,z+1)$ is satisfied.  Suppose $r<2n+1$, and let $x_0 = z, x_1, \ldots, x_k=z+1$ be witnesses for the existential quantifiers in $\psi_r(z,z+1)$.  For convenience (if $k<n$), let $x_j=z+1$ for $k < j\leq n$.  Note that $k=\lceil\frac{r}{2}\rceil\leq n$, so it is always the case that $x_n=z+1$.  

Now, as $\theta_2(z+1)$ holds, $\psi(u,z+1)$ is satisfied for some $u\in(z,z+1)$.  Let $x_1', \ldots, x_n' \in (z,z+1)$ be the witnesses for $\psi(u,z+1)$.  Let $m$ be the smallest index such that $x_m'<x_m$.  As $x_n'<z+1=x_n$ such an index must exist.  Then we claim that $x_1, \ldots, x_{m-1}, x_m', x_{m+1}', \ldots, x_n'$ are witnesses for $\psi(z,z+1)$.  Every interval $I_i$ defined by these witnesses\footnote{$I_{2k} = \{z_k\}$ and $I_{2k+1} = (z_{k},z_{k+1})$ where $z_i = x_i$ if $i<m$ and $x_i'$ if $i\geq m$}, except $I_{2m-1} = (x_{m-1},x_m')$, is either an interval defined by witnesses of $\psi_r(z,z+1)$ or an interval defined by witnesses of $\psi(u,z+1)$, so all points in $I_i$ satisfy $\varphi_i$ as required.  For the remaining interval, we observe that $(x_{m-1},x_m') \subseteq (x_{m-1},x_m)$, thus all points satisfy $\varphi_{2m-1}$ as required.  Thus $\psi(z)$ is satisfied.
\end{proof}
%and let 
%\[\begin{array}{rcll}
%\pi_k(z,z') &:=& \exists x_0 \exists x_1 \cdots \exists x_k\:& z =x_0 < x_1 < \cdots < x_k = z'\\
%&&&\wedge\: \forall y \in \{z\}. \varphi_0(y) \\
%&&&\wedge\: \forall y \in (z,x_1). \varphi_1(y) \\
%&&&\wedge\: \cdots\\
%&&&\wedge\: \forall y \in \{x_k\}. \varphi_{2k}(y).
%\end{array}\]

\end{document}